\documentclass[runningheads]{llncs}

\usepackage{amsmath}
\usepackage{amssymb}
\usepackage{xspace}
\usepackage{hyperref}

\usepackage[appendix=inline]{apxproof} 
\newtheoremrep{apxtheorem}[theorem]{Theorem}

\usepackage{todonotes}

\usepackage{tikz}
 \usetikzlibrary{trees}
 \usetikzlibrary{shapes}
 \usetikzlibrary{fit}
 \usetikzlibrary{shadows}
 \usetikzlibrary{backgrounds}
 \usetikzlibrary{arrows,automata}
 \usetikzlibrary{calc}

\tikzset{
   n/.style= {circle,fill,inner sep=1.5pt,node distance=2cm}
  ,acc/.style={circle,draw,inner sep=3pt,node distance=2cm}
  ,phantom/.style={circle},
  ,arr/.style={->, >=stealth, semithick, shorten <= 3pt, shorten >= 3pt}
}

\usepackage{xcolor}
\usepackage[notion,quotation]{knowledge}

\usepackage{quotes}
\definecolor{Blue Sapphire}{HTML}{005f73} 
\definecolor{Gamboge}{HTML}{ee9b00}
\definecolor{Ruby Red}{HTML}{9b2226}

\IfKnowledgePaperModeTF{
}{
	\knowledgestyle{intro notion}{color={Ruby Red}, emphasize}
	\knowledgestyle{notion}{color={Blue Sapphire}}
	\hypersetup{
		colorlinks=true,
		breaklinks=true,
		linkcolor={Blue Sapphire}, 
		citecolor={Blue Sapphire}, 
		filecolor={Blue Sapphire}, 
		urlcolor={Blue Sapphire},
	}
}
\IfKnowledgeCompositionModeTF{
	\knowledgeconfigure{anchor point color={Ruby Red}, anchor point shape=corner}
	\knowledgestyle{intro unknown}{color={Gamboge}, emphasize}
	\knowledgestyle{intro unknown cont}{color={Gamboge}, emphasize}
	\knowledgestyle{kl unknown}{color={Gamboge}}
	\knowledgestyle{kl unknown cont}{color={Gamboge}}
}{
}

 \knowledge{notion}
 | independence relation
 | independence relations
 | independent

 \knowledge{notion}
 | diamond

 \knowledge{notion}
 | communication architecture

 \knowledge{notion}
 | distributed alphabet

 \knowledge{notion}
 | asynchronous automaton
 | asynchronous automata
 | AA

 \knowledge{notion}
 | distributively recognized
 | recognizes distributively

 \knowledge{notion}
 | tree-like
 | TCA

 \knowledge{notion}
 | tree of $a$
 | $\tree_a$
 | \tree_a

 \knowledge{notion}
 | $\diam$
 | \diam

 \knowledge{notion}
 | $\diamtree$
 | \diamtree
 | $\diamtree(\tree_p)$

 \knowledge{notion}
 | view
 | views
 | $\traceview{X}{w}$
 | $\traceview{\tree}{w}$
 | $\traceview{\{p\}}{w}$
 | $\traceview{\{q\}}{w}$
 | \traceview{\{q\}}{w}
 | $\traceview{\{p\}}{w'}$

 \knowledge{notion}
 | shared parent view
 | $\traceparentview{p}{w}$
 | \traceparentview{p}{w}
 | $\traceparentview{p_i}{w}$

 \knowledge{notion}
 | \stateview{X}{w}
 | \stateview{\{p\}}{w}

 \knowledge{notion}
 | \stateparentview{p}{w}
 | \stateparentview{p'}{w}
 | \stateparentview{p}{w'}

 \knowledge{notion}
 | triangulated

 \knowledge{notion}
 | letter automaton

 \knowledge{notion}
 | maximal

 \knowledge{notion}
 | controllable

 \knowledge{notion}
 | uncontrollable

 \knowledge{notion}
 | controller
 | controllers

 \knowledge{notion}
 | $p$-aware
 | $\leaf$-aware
 | \leaf-aware

 \knowledge{notion}
 | $p$-memoryless
 | \leaf-memoryless

 \knowledge{notion}
 | winning
 | winning controller

 \knowledge{notion}
 | \leaf-short
 | \leaf-short automaton

 \knowledge{notion}
 | Control Problem

 \knowledge{notion}
 | \leaf-local strategy
 | \leaf-local strategies

 \knowledge{notion}
 | local actions
 | local action

\newcommand{\N}{\mathbb{N}}
\newcommand{\Proc}{\mathbb{P}}
\newcommand{\Alp}{\Sigma}

\newcommand{\diam}{\mathrm{Diam}}
\newcommand{\A}{\mathcal{A}}

\newcommand{\Lang}{\mathcal{L}}
\newcommand{\arch}{\ensuremath{\mathbb{C}}\xspace}

\newcommand{\CA}{\mathbb{C}}
\newcommand{\Arch}{\mathfrak{A}}
\newcommand{\Acc}{\mathit{Acc}}
\newcommand{\tree}{T}

\newcommand{\diamtree}{\mathrm{TDiam}}
\newcommand{\traceview}[2]{\mathrm{view}_{#1}(#2)}
\newcommand{\traceparentview}[2]{\mathrm{view}_{#1}^{\uparrow}(#2)}
\newcommand{\stateview}[2]{\sigma_{#1}(#2)}
\newcommand{\stateparentview}[2]{\sigma_{#1}^{\uparrow}(#2)}

\newcommand{\ifroot}[4]{\textsc{is-root}_{#1}^{#2}\textsc{?}{#3}\textsc{:}{#4}}
\newcommand{\invdef}[1]{(\texttt{I}_{\textsc{def}}[#1])}
\newcommand{\invstate}[1]{(\texttt{I}_S[#1])}
\newcommand{\invstateparent}[1]{(\texttt{I}_S^\uparrow[#1])}

\newcommand{\sys}{\mathit{sys}}
\newcommand{\env}{\mathit{env}}
\newcommand{\Runs}{\mathrm{Runs}}
\newcommand{\C}{\mathcal{C}}
\newcommand{\choice}{\mathit{ch}}
\newcommand{\leaf}{\ensuremath{\ell}\xspace}
\newcommand{\minusell}{{\ensuremath{\setminus \leaf}}}
\newcommand{\TS}{\mathit{TS}}
\renewcommand{\ts}{\mathit{ts}}

\begin{document}
\title{From Trees to Tree-Like: Distribution and Synthesis for Asynchronous Automata} 

\titlerunning{Distribution and Synthesis for Tree-Like Asynchronous Automata} 

\author{Mathieu Lehaut\inst{1}\orcidID{0000-0002-6205-0682}\thanks{Supported by Swedish research council (VR) project (No. 2020-04963) and NCN grant 2021/41/B/ST6/00535.} \and
Anca~Muscholl\inst{2}\orcidID{0000-0002-8214-204X}\thanks{Supported by ANR-23-CE48-0005 grant PaVeDyS.}$^\dagger$ \and
Nir~Piterman\inst{3}\orcidID{0000-0002-8242-5357}\thanks{Supported by Swedish research council (VR)
project (No. 2020-04963) and the Wallenberg AI, Autonomous Systems and Software Program (WASP) funded by the Knut and Alice Wallenberg Foundation.}\thanks{Part of this research was conducted while visiting the Simons Institute for the Theory of Computing.}
}

\institute{University of Warsaw, Poland \and
LaBRI, University of Bordeaux, France \and
University of Gothenburg and Chalmers University of Technology, Gothenburg, Sweden
}

\authorrunning{M. Lehaut, A. Muscholl, and N. Piterman} 



\maketitle 

\begin{abstract}
We revisit constructions for distribution and synthesis of Zielonka's asynchronous automata in restricted settings.
We show first a simple, quadratic, distribution construction for asynchronous automata, where the process architecture is tree-like. An architecture is  tree-like if there is an underlying spanning tree of the architecture and communications are local on the tree.
This quadratic distribution result generalizes the known construction for tree
architectures and improves on an older, exponential construction for triangulated dependence alphabets. 
Lastly we consider the problem of distributed controller synthesis and show that it is decidable for tree-like architectures.
This extends the decidability boundary from tree architectures to tree-like keeping the same $\text{Tower}_d(n)$ complexity bound, where $n$ is the size of the system and $d \ge 0$ the depth of the
process tree.

\end{abstract}

\section{Introduction}
\label{sec:intro}

We consider a canonical formalism  for distributed systems
with a fixed communication structure~--~Zielonka's \emph{asynchronous automata}.
These are a well-known model that supports distributed synthesis
under a fixed communication topology, rooted in the theory of Mazurkiewicz
traces~\cite{Maz77}. In Zielonka automata  processes are
connected to a specific set of channels, each process to their own subset of the
channels. Processes communicate by synchronizing on channels they are connected
to. During communication, all involved processes share their local states and
then, based on this mutually shared information, move to new states. In
addition, this model is rendez-vous: communication occurs only if all participants
(processes connected to the channel) agree to it. The model is asynchronous as
communications on two channels that do not have a process in common can be
performed in parallel. Highlighting the importance of this model is Zielonka’s seminal
result about distribution of regular languages. Given a communication
architecture – which process is connected to which channels – and a regular
language that respects the asynchrony of the architecture, it is always possible
to distribute this language into an asynchronous automaton
\cite{Zielonka87}.

Zielonka’s result is a prominent, and rare, example of distributed
synthesis, yet computationally expensive. Starting with a regular language accepted by
an automaton with $n$ states and a topology with $p$ processes, an equivalent
asynchronous automaton of size $O(4^{p^4}n^{p^2})$ can be
constructed~\cite{GenestGMW10}. Note that the exponential is only due to the
number of processes (which is part of the input), but this is indeed needed, as
shown by a lower bound in~\cite{GenestGMW10}.  The
construction
is quite involved and, even after nearly 40 years, it is not widely or easily
understood. It has been the source for much research on how to understand,
explain, and improve it (e.g.~\cite{MukundS97,genest2006constructing,GenestGMW10,adgs13}). Recently, an alternative construction for asynchronous automata
has been proposed using a partial-order variant of Propositional Dynamic
Logic~\cite{AGKW24}.

The complexity of the general Zielonka result prompted researchers to look at
simpler cases where distribution could be easier. One notable example is the
construction of Krishna and Muscholl~\cite{KrishnaM13} who show that when the communication
architecture is restricted to a tree – every channel is listened to by at most two
processes and there are no cycles – then every process needs a quadratic number
of states in that of the original sequential automaton.  Muscholl
and Diekert~\cite{DiekertM96} showed a simpler construction, however exponential in the number
of states of the sequential automaton, for triangulated dependence alphabets.

The asynchronous automata model is notable also for its usage in control, similar to Ramadge and Wonham's
supervisory control of discrete event systems~\cite{RW87,LRT18}. 
For systems communicating synchronously by shared variables, it has been
established early on that distributed control is undecidable \cite{PR90}, and
only decidable for very restricted communication architectures, moreover with very
high complexity \cite{KV01,FS05}.
In contrast, for distributed control based on asynchronous communication using Zielonka automata
researchers have found more complex architectures for which the problem was
still decidable.
Notably, control for $\omega$-regular
winning conditions was shown to be decidable when communication is restricted to a tree
\cite{muscholl2014distributed} (see also~\cite{GGMW13} for local reachability conditions). 
Their construction has $\text{Tower}(d)$ complexity in the depth $d$ of the
tree, but works for an architecture that is undecidable in the synchronous
shared-variable world.
Similarly, for \emph{uniformly well-connected} architectures, the problem is
again decidable but this time in exponential space \cite{GastinSZ09}. On a
general note, distributed control based on Zielonka automata, also known as
asynchronous games with causal memory, is equivalent to so-called Petri games~
\cite{Finkbeiner19concur}.
However, in spite of earlier hopes emanating from these general architectures having a decidable control problem, the general distributed synthesis has recently been established as undecidable \cite{Gim22}. 

Recently, Hausmann et~al. studied distribution in the model of reconfigurable asynchronous automata \cite{HausmannLP24}.
Their model generalizes asynchronous automata by allowing processes to change the set of channels they communicate on at runtime. 
They established that the simple distribution construction of Krishna and Muscholl \cite{KrishnaM13} is applicable also to a more general architecture they call \emph{tree-like}.
That is, processes are allowed to communicate more freely as long as their communication has an underlying spanning tree and communications are local on the tree.
Here we present their construction in the context of asynchronous automata removing the complex notations related to the reconfiguration.
While their construction was exponential in the number of processes in order to follow the structure of the tree, when the communication architecture is fixed this is no longer required and the quadratic construction of Krishna and Muscholl emerges.
Furthermore, we show that this construction also improves a previously known exponential
construction by Diekert and Muscholl in the context of triangulated dependence graphs
\cite{DiekertM96}, as they turn out to be equivalent to tree-like architectures.
Finally, we revisit the distributed synthesis result for tree architectures
\cite{muscholl2014distributed} and show that this construction as well can be
generalized to the case of tree-like architectures.

\smallskip

For convenience, technical terms and notations in the
electronic version of this manuscript are hyper-linked to their 
definitions (cf.~\url{https://ctan.org/pkg/knowledge}).

\section{Preliminaries}
\label{sec:prelim}

\subsection{Deterministic Finite Automata}

A deterministic finite automaton (DFA) over alphabet $\Sigma$ 
is denoted as $\A = (\Sigma,S,\Delta,s_0, F)$, where
$S$ is the finite set of states, 
$\Delta:S\times \Sigma \to S$ the partial transition
function, $s_0\in S$ the initial state,
and $F\subseteq S$ the set of accepting states.
Given a word $w=a_0\cdots a_{n-1}$, an initial run of $A$ on $w$ is a path
$s_0 \stackrel{a_0}{\to} s_1 \stackrel{a_1}{\to} \ldots \stackrel{a_{n-1}}{\to}
s_n$ of $\A$, i.e.~for every $0\leq i < n $ we have $s_{i+1} = \Delta(s_i,a_i)$.
An initial run is accepting if $s_n \in F$ and then $w$ is accepted by
$\A$.
The language of $\A$, denoted by $\Lang(\A)$, is the set of
words accepted by $\A$. Often we abuse notation and write $\Delta(s,u)$ for the
state $s'$ reached from $s$ on the
word $u$ (or just $s \stackrel{u}{\to} s'$).

An \intro{independence relation} is a symmetric, irreflexive relation $I \subseteq
\Sigma\times \Sigma$.
Two words $u, v \in \Sigma^*$ are said
to be $I$-indistinguishable, denoted by $u \sim_I v$, if one can start
from $u$, repeatedly switch two consecutive independent letters, and
end up with $v$.
That is, $\sim_I$ is the transitive closure of the relation $\{(uabv,ubav) \mid (a,b)\in I, \, u,v\in \Sigma^*\}$.
We denote by $[u]_I$ the $I$-equivalence class of a word $u\in\Sigma^*$. This is
a.k.a.~a Mazurkiewicz trace~\cite{Maz77}.
Let $\A = (\Sigma,S,\Delta,s_0,F)$ be a deterministic automaton over
$\Sigma$.
We say that $\A$ is $I$-\intro{diamond} if for all pairs of independent
letters $(a,b) \in I$ and all states $s \in S$, we have $\Delta(s,ab)
= \Delta(s,ba)$.
If $\A$ has this property, then a word $u$ is accepted by $\A$ if and
only if all words in $[u]_I$ are accepted.

\subsection{Asynchronous Automata}
A \intro{communication architecture} 
$\arch: \Sigma \to 2^\Proc$ associates with each letter the
subset of processes reading it.
We put $\arch^{-1}(p) = \{a \in \Sigma \mid p \in \arch(a)\}$.
A \intro{distributed alphabet} is $(\Sigma,\arch)$, where $\arch$ is a
communication architecture.
It induces an \kl{independence relation}
$I(\arch) \subseteq \Sigma\times\Sigma$ by $(a,b)\in I(\arch)$ iff
$\arch(a)\cap
\arch(b)=\emptyset$.
The complement of the independence relation is a dependence relation $D(\arch) = \Sigma \times \Sigma \setminus I(\arch)$.
It is simple to see that $(a,b)\in D(\arch)$ iff $\arch(a)\cap \arch(b)\neq \emptyset$.
A dependence relation $D$ induces  a graph $G_D=(\Sigma,D)$, where $\Sigma$ is the set of nodes and $D$ is the set of edges. 
We say that $\arch$ is binary if for each $a\in \Sigma$ we have
$|\arch(a)|\le 2$.
A binary $\arch$ induces a tree if the graph $(\Proc,E)$, where
$(p,q)\in E$ iff $\arch(a)=\{p,q\}$ with $p \not= q$ for some $a\in \Sigma$, is a tree.
\AP For any process $p \in \Proc$, 
let $\Alp_p = \{a \in \Alp \mid \{p\}=
\arch(a)\}$ be the set of $p$-\intro{local actions}, that is the set
of actions involving only $p$.

\AP An \intro{asynchronous automaton} (in short: \reintro{AA}) \cite{Zielonka87} over a
\kl{distributed alphabet} $(\Sigma,\arch)$ and processes $\Proc$ is a tuple
$\mathcal{B} = ((S_p)_{p \in \Proc}, (s^0_p)_{p \in \Proc}$,
$(\delta_a)_{a \in \Alp}, \Acc)$
such that:
\begin{itemize}
	\item $S_p$ is the finite set of states for process $p$, and $s^0_p
	\in S_p$ is its initial state,
	\item $\delta_a: \prod_{p \in \arch(a)} S_p \to \prod_{p \in
	\arch(a)}
	S_p$ is a partial transition function for letter $a$ that
	only depends on the states of processes in $\arch(a)$ and changes
	them,
	\item $\Acc \subseteq \prod_{p \in \Proc} S_p$ is a set of accepting
	states.
\end{itemize}
A global state of $\mathcal{B}$ is  $\textbf{s} = (s_p)_{p
\in \Proc}$, giving the state of each process.
For a global state $\textbf{s}$ and a subset $P \subseteq \Proc$, we
denote by $\textbf{s}\downarrow_P = (s_p)_{p \in P}$ the part of
$\textbf{s}$ consisting of states from processes in $P$.

An initial run of $\mathcal{B}$ on a word $a_1 a_2\ldots a_n$ is a path $\textbf{s}_0 \stackrel{a_1}{\to}
\textbf{s}_1 \stackrel{a_2}{\to}
\dots \stackrel{a_n}{\to} \textbf{s}_n$ of the product (global) automaton from
the initial state $\textbf{s}_0 =
(s^0_p)_{p \in \Proc}$: for all $0 < i \leq n$, $\textbf{s}_i \in
\prod_{p \in \Proc} S_p$, $a_i \in \Alp$, satisfying 
$\textbf{s}_{i}\downarrow_{\arch(a_i)} =
\delta_{a_i}(\textbf{s}_{i-1}\downarrow_{\arch(a_i)})$
 and
$\textbf{s}_{i}\downarrow_{\Proc\setminus\arch(a_i)} =
\textbf{s}_{i-1}\downarrow_{\Proc\setminus\arch(a_i)}
$.
An initial run is accepting if $\textbf{s}_n$ belongs to $\Acc$.
The word $a_1 a_2 \dots a_n$ is accepted by $\mathcal{B}$ if such an
accepting run
exists (note that automata are deterministic but runs on certain words
may not exist).
The language of $\mathcal{B}$, denoted by $\Lang(\mathcal{B})$, is the
set of words
accepted by $\mathcal{B}$.
Let $\Runs(\mathcal{B})$ denote the set of runs of $\mathcal{B}$, and $\Runs_p(\mathcal{B})$ their projection on the $p$-component of the state.

\AP We say that a language $\Lang \subseteq \Sigma^*$ over
$(\Sigma,\arch)$ is \intro{distributively recognized} if
there exists an \kl{asynchronous automaton} $\mathcal{B}$
such that $\Lang(\mathcal{B})=\Lang$.

\begin{theorem}[\cite{Zielonka87,GenestGMW10,KrishnaM13}]
	Given an $I(\arch)$-\kl{diamond} deterministic automaton $\A$, there
	exists an \kl{AA} $\mathcal{B}$
	that distributively recognizes the language of $\A$.
	In general, if $\A$ has $n$ states then every process of $\mathcal{B}$ has
	$O(4^{|\Proc|^4} n^{|\Proc|^2})$ states.
	If $\arch$ induces a tree, then every process of $\mathcal{B}$ has $O(n^2)$
	states.
\end{theorem}

The size of of an \kl{AA} $\mathcal{B}$ (written as $|\mathcal{B}|$) is defined as $\max_{p\in\Proc} |S_p|$.
This standard notion of size takes into account the maximal local memory (state space) of
a process, because an \kl{AA} is a \emph{distributed} device. 

\subsection{Tree-like communication architectures}

As before, let $\Proc$ be a set of processes and $\Sigma$ an alphabet, both finite.
A communication architecture is tree-like if there is a tree that spans
letters that synchronize more than two processes.
Formally, we have the following.

Let $|\Proc| = n$.
A \emph{tree} over $\Proc$ is a  
$\tree = (\Proc,r,E)$ 
where $\Proc$ is the set of nodes, of which $r
\in \Proc$ is the \emph{root}, and $E \subseteq \Proc^2$ is the set of
\emph{edges} that is cycle free and connected.

\begin{definition}\label{def:TCA}
We say that $\Arch = (\CA,\tree)$, where $\CA$ is a communication
architecture and $\tree$ is a tree, is a \intro{tree-like} \kl{communication
architecture} (short: \kl{TCA}) if
\begin{enumerate}
\item\label{item:channelconnectivity} for all letters $a\in\Sigma$, the
set $\arch(a)$ is connected in $T$, i.e. if $p,q \in \arch(a)$, then
all processes along the path from $p$ to $q$ in $T$ are also in
$\arch(a)$,
and
\item\label{item:edgechannel} if $(p,q) \in E$, then there is a letter 
$a \in \Sigma$ such that $p,q \in \arch(a)$.
\end{enumerate}
\end{definition}
Condition~\ref{item:channelconnectivity} ensures that a communication on a given letter is always ``local'' in the tree.
Condition~\ref{item:edgechannel} ensures that the tree does not consist of disconnected parts.
In case that condition~2 does not hold we say that $\Arch$ is \emph{forest-like}.

\begin{example}
Fix sets $\Proc = \{p_1,\dots,p_5\}$ of processes and $\Sigma = \{a_1, a_2, a_3\}$ of letters.
Then $\Arch = ((\arch(a_1), \arch(a_2), \arch(a_3)), \tree)$ as given
in Figure~\ref{fig:archi} is a \kl{tree-like} communication architecture
rooted in $p_1$:
\begin{enumerate}
\item Every letter is connected in $\tree$.
\item All edges in $\tree$ are covered by at least one letter.
\end{enumerate}

\end{example}

\begin{figure}[bt]
\tikzset{every state/.style={minimum size=15pt}}
\begin{center}
  \begin{tikzpicture}[
		auto,
    node distance=1.5cm,
    semithick
    ]
    \node[state] (1) {$p_1$};
    \node[state,below left = of 1] (2) {$p_2$};
    \node[state,below right = of 1] (3) {$p_3$};
    \node[state,below left = of 3] (4) {$p_4$};
    \node[state,below right = of 3] (5) {$p_5$};


    \path[->] (1) edge (2);
    \path[->] (1) edge (3);
    \path[->] (3) edge (4);
    \path[->] (3) edge (5);

    \draw[rounded corners,color=red] ($(1.north)+(0,0.4)$) -- ($(2.west)+(-0.4,0)$) -- ($(2.south)+(0,-0.4)$) -- ($(1.south)+(0,-0.4)$) -- ($(3.south)+(0,-0.4)$) -- ($(3.east)+(0.4,0)$) -- cycle node[above,color=red] {$a_1$};
    \draw[rounded corners,color=blue] ($(1.west)+(-0.3,0)$) -- ($(3.west)+(-0.3,0)$) -- ($(4.west)+(-0.3,0)$) node[left,color=blue] {$a_2$} -- ($(4.south)+(0,-0.3)$) -- ($(3.east)+(0.3,0)$) -- ($(1.north)+(0,0.3)$) -- cycle;
    \draw[rounded corners,color=teal] ($(3.north)+(0,0.2)$) -- ($(3.west)+(-0.2,0)$) -- ($(5.south)+(0,-0.2)$) -- ($(5.east)+(0.2,0)$) node[right,color=teal] {$a_3$} -- cycle;
  \end{tikzpicture}
\caption{A tree-like communication architecture: the  tree is
given by the black edges, while the communication architecture is drawn
with one color group for each letter. 
}\label{fig:archi}
\end{center}
\vspace*{-1cm}
\end{figure}
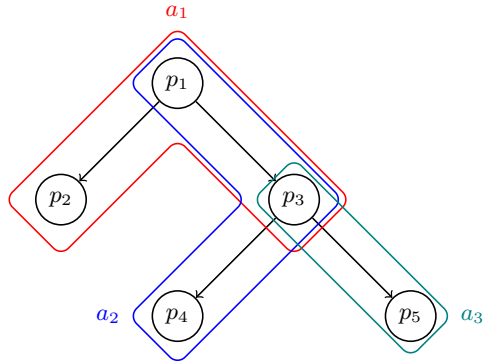

Notions of parents, children, descendants, leaves, neighbors, and subtrees are defined as usual.

\section{Asynchronous Automata with Tree-like Architectures}
\label{sec:async automata}

Hausmann et al.~\cite{HausmannLP24} showed that \emph{reconfigurable asynchronous automata} can be distributed when they communicate over tree-like architectures.
Reconfigurable asynchronous automata change their communication architectures during the operation. 
Thus, they are more general than asynchronous automata.
Here, we present their construction for the case of a \emph{fixed} \kl{tree-like} architecture.
Equivalently, we show that the construction of Krishna and
Muscholl~\cite{KrishnaM13} for tree architectures extends, with more complicated
transitions, to
tree-like architectures.


We extend the definition of independence to sequences and to sets of letters.
First we set $C_1,C_2\subseteq \Alp$ as \kl{independent} (and write ($C_1,C_2)
\in I(\CA)$) if for every $a_1\in C_1$ and $a_2\in C_2$ we have $(a_1,a_2)\in I(\CA)$.
Equivalently, $C_1\times C_2 \subseteq I(\CA)$. 
We set $w_1,w_2\in\Alp^*$ as \kl{independent} (and write $(w_1,w_2) \in I(\CA)$) if there exist sets $C_1,C_2\subseteq \Alp$ such that $w_1\in C_1^*$, $w_2\in C_2^*$, and $(C_1,C_2)\in I(\CA)$. 

Fix a DFA $\A = (\Alp, S, \Delta, s_0,F)$ and  
$\Arch = (\CA,\tree)$ such that $\A$ is $I(\CA)$-diamond.
If $w_1$ and $w_2$ are two $I(\CA)$-\kl{independent} words, then we have that $\Delta(s, w_1 w_2) = \Delta(s, w_2 w_1) = \Delta(s, w)$ for any $w$ that is an interleaving of $w_1$ and $w_2$.
This can be verified by induction on the length of $w$.

We now study what information is needed in order to compute $\Delta(s,w_1 w_2)$ without memorizing $w_1$ and $w_2$ themselves.


Fix a state $s$ of $\A$ and let $w_1$ and $w_2$ be two $I(\CA)$-\kl{independent} words.
Let $s_1 = \Delta(s,w_1)$, $s_2 = \Delta(s,w_2)$, and let $C_1,C_2$ be sets such that $w_1\in C_1^*$, $w_2\in C_2^*$ and $(C_1,C_2)\in I(\CA)$.
The next lemma shows that the state $\Delta(s,w_1w_2)$ can be computed knowing
only $s$, $s_1$, $s_2$, and $C_2$ (but not the exact words $w_1,w_2$).

\begin{lemma}[\cite{CMZ93}]\label{lem:diam}
Let $\A$ be an $I(\CA)$-\kl{diamond} DFA. There exists a (partially-defined) function \intro{$\diam$} $: S^3 \times 2^{\Alp} \to S$ such that if $s, s_1, w_1, s_2, w_2, C_2$ are all defined as described above, then $\diam(s,s_1,s_2,C_2) = \Delta(s,w_1w_2)$.
\end{lemma}

\begin{proof}
Let $C_1'$ be the maximal set of letters such that $(C_1',C_2)\in I(\CA)$. It
follows that $C_1 \subseteq C_1'$. 
Consider some arbitrary words $w_1'$ and $w_2'$ such that $w_1'\in C_1'^*$, $w_2' \in C_2^*$ and such that $s_i = \Delta(s,w_i)=\Delta(s, w_i')$ for $i \in \{1,2\}$.
Let us show that $\Delta(s,w'_1w_2')=\Delta(s,w_1w_2)$.

First, by diamond property and hypothesis we have that $s':=\Delta(s,w_1w_2) =
\Delta(s, w_2w_1)=\Delta(s_2,w_1)$.
We also know that $\Delta(s,w_2') = s_2$, so $s' = \Delta(s,w_2' w_1)$.
Then, reusing the diamond property: $s' = \Delta(s,w_2' w_1) = \Delta(s,w_1 w_2')$.
We conclude reusing the hypothesis: $s'=\Delta(s,w_1 w_2') = \Delta(s,w'_1w_2')$.

Finally we can set $\diam(s,s_1,s_2,C_2)$ as the state $\Delta(s,w'_1w_2')$,
for some $w'_1 \in C_1'^*,w'_2 \in C_2^*$.
\qed
\end{proof}

In \cite{KrishnaM13}, the $\diam$ function has been used for \kl{TCA} 
$\Arch = (\CA,\tree)$ with binary $\CA$ and $I(\CA)$-diamond DFA $\A$ in the following way.
On a communication between a parent and one of its children, let $s$ be the most
recent 
state of $\A$ known by both the parent and the child, $s_1$ and $s_2$ be the most recent states
known by the parent and the child respectively, and $C_2$ the set of letters
with domains in the subtree rooted in the child.
Then one can combine the information known by the parent and the child after the
communication by computing $\diam(s,s_1,s_2,C_2)$.
We do not provide more details because we show below the more general case of
non-binary $\CA$.

In our \kl{tree-like} setting, communications can include more than two participants.
To that end, we naturally extend the $\diam$ function to work on subtrees.
Let $\tree$ be a tree where each node is labeled by a pair $(s,t)$ of states and
a set $C^\downarrow$ of letters, with the intuition that $s$ is the most recent state of $\A$
known by both the process at this node  and its parent, $t$ is 
the most recent state of $\A$ known by
this process, and $C^\downarrow$ is the set of letters that can be
found in the subtree rooted in that node and that are \kl{independent} from those shared with the
parent.
Let $r$ be the root of $\tree$, labeled by $(s_r,t_r,C_r^\downarrow)$, and assume that $r$ has $k$ children labeled $(s_i,t_i,C_i^\downarrow)$ and leading to tree $\tree_i$ for $i \in \{1,\dots,k\}$.
\AP Then we define the function \intro{$\diamtree$} recursively in the following way:
\begin{align*}
\diamtree(\tree) =
\begin{cases}
t_r &\text{ if $k = 0$,}\\
\kl{\diam}(s_k, \diamtree(\tree \setminus \tree_k), \diamtree(\tree_k), C_k^\downarrow) &\text{ otherwise.}
\end{cases}
\end{align*}
Note that the $\mathrm{state}(p)$ function from \cite{KrishnaM13} is equivalent
to \kl{$\diamtree(\tree_p)$} where $\tree_p$ is the tree rooted in $p$.

Consider a letter $a\in \Alp$, where $\CA(a)=\{p_1,\ldots, p_k\}$. 
\AP We define the \intro{tree of $a$}, denoted by \reintro{$\tree_a$}, as the subtree of $\tree$ whose set of nodes is exactly $\CA(a)$; by definition of a \kl{TCA} this is indeed a tree.
Consider now a second letter $b \neq a$ and some $p$ such that $p \in \CA(a)$ but $p \notin \CA(b)$. 
Again by definition of a \kl{TCA}, $\CA(b)$ occupies a connected area of $\tree$, and this area does not contain $p$. 
Therefore, it must either lie "above" $p$, meaning that the (unique) path from $p$ to this area goes through the parent of $p$, or "below" $p$, meaning that it goes through one of $p$'s children.
We define $C_p^\downarrow$ to be the set of letters for which the second case holds.
Now assume that each $p_i \in \CA(a)$ is associated with a pair of states
$(s_i,t_i)$ of $\A$.
Then we label each node $p_i$ of $\kl{\tree_a}$ by $(s_i,t_i,C_{p_i}^\downarrow)$ and define the state $s_a = \diamtree(\tree_a)$.
Intuitively, that means we combine the information of all nodes in $\tree_a$ to compute the most up-to-date state $s_a$.

We are now ready to state our construction.
We distribute $\A = (\Alp, S, \Delta, s_0, F)$ by defining $\mathcal{B} = ((S_p)_{p \in \Proc}, (s^0_p)_{p \in \Proc},(\delta_a)_{a \in \Alp},\Acc)$.
Namely, for each $p\in \Proc$ we define the set $S_p$ and its initial state
$s^0_p$, the  transition function $\delta_a$ for each $a\in \Alp$ and the
acceptance set $\Acc$. 
For each $p\in \Proc$ we set $S_p=S\times S$ and $s^0_p=(s_0,s_0)$.

For $a\in \Alp$ we describe now the transition $\delta_a: \prod_{p\in \CA(a)} S_p \rightarrow \prod_{p\in \CA(a)} S_p$.
The transition computes the diamond closure over the states held by all other processes in the correct order. 
Based on this computation each process updates its state.

Let $\CA(a)=\{p_1,\ldots, p_k\}$ and for all $1 \leq i \leq k$ let $(s_i,t_i)$ be the state of $p_i$.
Let $s_a = \kl{\diamtree}(\kl{\tree_a})$ as described earlier, and let $s' = \Delta(s_a,a)$.
We set 
$$\delta_a((s_1,t_1), \dots, (s_k,t_k))=((u_1,v_1), \dots, (u_k,v_k))$$
where the new state of process $p_i$ is $u_i=\ifroot{a}{p_i}{s_i}{s'}$ and $v_i=s'$ with 
\[
\ifroot{a}{p}{s}{t} = 
\left \{
\begin{array}{l r}
		s &\text{ if $p$ is the root of \kl{$\tree_a$},}\\
		t &\text{ otherwise.}
\end{array}
\right .
\]
If $\Delta(s_a,a)$ is undefined, then so is $\delta_a$.

We now define $\Acc$.
Consider a global state $\textbf{s}=((s_1,t_1),\ldots (s_n,t_n))$.
Recall that the \kl{TCA} is $\Arch = (\CA,\tree)$.
For the tree $\tree$, consider the labeling of every node $i$ in $\tree$ by the pair $(s_i,t_i)$ and the set $C_i^\downarrow$ of letters found "below" node $i$, as defined earlier.
We then set $\textbf{s}\in \Acc$ if and only if $\diamtree(\tree)\in F$.

We state our main result below. 

\begin{theorem}\label{thm:correctness}
If $\A$ is $\CA$-\kl{diamond} and $\CA$ is \kl{tree-like}, then $\mathcal{B}$
\kl{recognizes distributively} $\Lang(\A)$.
The size of $\mathcal{B}$ is in $O(n^2)$, where $n$ is the number
of states of $\A$.
\end{theorem}


\subsection{Proof of correctness}\label{sec:invariants}
This section is dedicated to the proof of Theorem~\ref{thm:correctness}.

First, and completely independently from the construction, we need to define
what is the view  of a (set of) process(es), as in \cite{KrishnaM13},
and then show a useful property relating \kl{$\diamtree$} and those views.
Let $\A$ be a $\CA$-diamond DFA, $w \in \Alp^\ast$  a word, $\rho = s_0
\xrightarrow{a_1} \dots \xrightarrow{a_k} s_k$ the initial run of $\A$ on $w=a_1 \dots
a_k$,  and
$X \subseteq \Proc$ a set of processes.
As processes in $X$ are not necessarily part of every communication in $w$,
there may be some parts of $\rho$ that are happening outside of $X$'s knowledge
and processes in $X$ may not know the state reached at the end of $\rho$
just by themselves.
\AP We define recursively the \intro{view} 
of $X$ on $w$, which intuitively corresponds to
the subword of $w$ that processes in $X$ can see through shared actions,
and denote it by
\reintro{$\traceview{X}{w}$}: 
\begin{align*}
&\traceview{X}{\varepsilon} = \varepsilon\\
&\traceview{X}{w \cdot a_k} =
\begin{cases}
\traceview{X \cup \CA(a_k)}{w} \cdot a_k &\text{if $X \cap \CA(a_k) \neq \emptyset$,}\\
\traceview{X}{w} &\text{otherwise.}
\end{cases}
\end{align*}
We compute what is also known as the causal past of processes in $X$:
We start from the last communication in which at least one process in
$X$ participated. Then, the $\traceview{X}{w}$ computes backwards the subsequence of $w$
seen by processes in $X$ through shared actions.
Abusing notations, we identify any tree $\tree = (\Proc,r,E)$ with its set of
processes $\Proc$ and write \kl{$\traceview{\tree}{w}$} instead of
$\traceview{\Proc}{w}$.
\AP We also define the state of $\A$ that those processes can compute as $\intro{\stateview{X}{w}} = \Delta(s_0,\traceview{X}{w})$.
Furthermore, with $\A$ being $\CA$-diamond, we can deduce that for any letter
$a$, $\Delta(s_k,a)$ is defined if and only if $\Delta(s_0,
\traceview{\arch(a)}{w}a)$ is defined.
That is, whether a communication with letter $a$ can happen cannot depend on communications made outside of the \kl{view} of those listening to $a$.

\AP Let  $p \in \Proc$ be some process with parent $q$ in $\tree$. The \intro{shared
parent view} of $p$ on $w$, 
denoted by \reintro{$\traceparentview{p}{w}$}, is defined as  \kl{$\traceview{\{p\}}{w'}$} where $w'$ is
the minimal prefix of $w$ containing all occurrences of all letters
$b$ such that $\{p,q\} \subseteq \CA(b)$.
In other words this is the sequence of actions, as seen from $p$'s point of view,
until the last communication involving both $p$ and its parent.
This is not necessarily the same as \kl{$\traceview{\{p\}}{w}$}, because $p$ may have 
communicated with its children after the end of $w'$.
For similar reasons, \kl{$\traceparentview{p}{w}$} is not necessarily the same as \kl{$\traceview{\{q\}}{w}$}.
If $p$ is the root of $\tree$, we set $\traceparentview{p}{w} = \varepsilon$ for all $w$.
\AP Finally, we set $\intro{\stateparentview{p}{w}} =
\Delta(s_0,\kl{\traceparentview{p}{w}})$ to be the state reached in $\A$ after
reading that sequence.

The \kl{$\diam$} function, and by extension \kl{$\diamtree$}, interplay nicely with the views of independent parts of the system.
More specifically, we show that if each node in $\tree$ is correctly labeled, then $\diamtree(\tree)$ correctly computes the view according to every process in $\tree$.

\begin{lemma}\label{lemma:diamview}
For all words $w$ 
and for all subtrees $\tree'$ of $\tree$ where each node $p$ of $\tree'$ is
labeled by $(\kl{\stateparentview{p}{w}}, \stateview{\{p\}}{w}, C_p^\downarrow)$, we
have $\kl{\diamtree}(\tree') = \stateview{\tree'}{w}$.
\end{lemma}

\begin{proof}
By induction on the structure of $\tree'$, same as in \cite{KrishnaM13}. 

If $\tree'$ is a single node $p$ labeled by $(s_p,t_p,C_p^\downarrow)$ then $\diamtree(\tree') = t_p = \kl{\stateview{\{p\}}{w}}$ by assumption.
Otherwise, assume $\tree'$ is made of a root $p$ labeled by $(s_p,
t_p,C_p^\downarrow)$ with $k$ children $p_1,\dots, p_k$  labeled by $(s_1,t_1,C_1^\downarrow),
\dots, (s_k,t_k,C_k^\downarrow)$ and leading to subtrees $\tree_1, \dots,
\tree_k$, respectively.
Assume that the property holds on those subtrees, i.e. that for all $i \leq k$,
we have $\diamtree(\tree_i) = \stateview{\tree_i}{w}$.
Then with $\tree_i'$ denoting the subtree of $\tree'$ composed of $p$, its
first $i$ children, and their respective descendants, we recursively show that
\[
(P_i)\qquad \qquad  \diamtree(\tree_i') = \stateview{\tree_i'}{w} 
\]
For the initial step $i = 0$, $\diamtree(\tree_0') = t_p = \stateview{\{p\}}{w}$ by assumption.
Now suppose $(P_{i-1})$ holds for $1 \leq i \leq k$.
By construction 
$$\diamtree(\tree_i') = \diam(s_i,\diamtree(\tree_{i-1}'), \diamtree(\tree_i), C_i^\downarrow).$$ 
By $(P_{i-1})$ we know that  $\diamtree(\tree_{i-1}') = \stateview{\tree_{i-1}'}{w}$, and by assumption we have that $s_i= \stateparentview{p_i}{w}$ and $\diamtree(\tree_i) = \stateview{\tree_i}{w}$.

Let $w_0 =$ \kl{$\traceparentview{p_i}{w}$}, $w_i = \traceview{\tree_i}{w}$, and $w_{i-1} = \traceview{\tree_{i-1}'}{w}$.
By definition, $w_0$ is a prefix of both $w_i$ and $w_{i-1}$.
Thus, let $w_i = w_0 \cdot w_i'$ and $w_{i-1} = w_0 \cdot w_{i-1}'$.
As $w_0$ contains the last communication involving $p_i$ and its parent $p$, it
is easy to see that $w_i'$ and $w_{i-1}'$ are independent, otherwise $w_0$ would
not be their longest common prefix.
Note also that $\traceview{\tree'_i}{w}=w_0 v$, where $v$ is an
interleaving of $w_i'$ and $w_{i-1}'$.
Also, all actions of $w_i'$ are in $C_{p_i}^\downarrow$, otherwise they would involve $p$.
Using the definition of \kl{$\diamtree$}, we get
$\diamtree(\tree'_i)=\diam(s_i,\stateview{\tree_{i-1}'}{w},
\stateview{\tree_i}{w},C_i^\downarrow)$ with $s_i=\stateparentview{p_i}{w}=\Delta(s_0,w_0)$,
$\stateview{\tree_{i-1}'}{w}=\Delta(s_i,w_{i-1}')$ and
$\stateview{\tree_i}{w}=\Delta(s_i,w_i')$. Thus, by Lemma~\ref{lem:diam}, $\diamtree(\tree'_i)=\stateview{\tree'_i}{w}$.

We have shown that $(P_i)$ holds for all $i$, and therefore for $i = k$, and as $\tree'= \tree_k'$ we have $\diamtree(\tree') = \stateview{\tree'}{w}$.
\qed
\end{proof}

Note that this lemma implies that if $\tree'$ is the full tree $\tree$ and is correctly labeled, then $\diamtree(\tree) = \stateview{\Proc}{w} = \Delta(s_0,w)$.
We will use this property to prove the correctness of our construction.

Let us define invariants that should be satisfied throughout a run.
Let $w\in\Alp^*$.
Our first invariant states that a run $\rho'$ of $\mathcal{B}$ on $w$ exists if and only if a run $\rho$ of $\A$ on the same word $w$ exists.
\[
\invdef{w}: \text{The run $\rho$ on $w$ is defined $\Leftrightarrow$ The run $\rho'$ on $w$ is defined.}
\]
Now if both runs are undefined, then $w$ is trivially rejected by both $\A$ and $\mathcal{B}$ and there is nothing left to show.
So for the remaining invariants, assume that both $\rho$ and $\rho'$ are defined and end in states $s$ and $\mathbf{s} = ((s_1, t_1), \dots, (s_n,t_n))$ respectively.
Then we establish the invariants relating the actual state of $\rho$ with the pair of states that each process keeps in its local state, which corresponds to the intuition given in the previous section.
\begin{align*}
&\invstateparent{w}: \quad \forall p \in \Proc.\, s_p = \kl{\stateparentview{p}{w}}
&\invstate{w}: \quad  \forall p \in \Proc.\, t_p = \kl{\stateview{\{p\}}{w}}
\end{align*}
Those two invariants, used with Lemma~\ref{lemma:diamview}, are key to prove
that the languages of $\A$ and $\mathcal{B}$ are equivalent.
Let us now prove that those invariants hold inductively.

\medskip

\noindent \emph{Initialization of the invariants.}
All invariants hold on the empty word $w = \varepsilon$.
$\invdef{\varepsilon}$ is trivial.
$\invstateparent{\varepsilon}, \invstate{\varepsilon}$ are easily derived from
the definition of the initial state $\mathbf{s}_0$ of ${\cal B}$.

\medskip

\noindent \emph{Invariants after a transition.}
Now consider a word of the form $w' = w \cdot a$, with both runs $\rho$ and $\rho'$ defined up to $w$ ending in states $s$ and $\mathbf{s}$ respectively.
Assume that all invariants hold in $w$.
We show that they still hold in $w'$.

For $\invdef{w'}$, by using $\invstateparent{w}$, $\invstate{w}$, and Lemma~\ref{lemma:diamview}, we have that $s' := \diamtree(\tree_a) = \stateview{\CA(a)}{w}$.
Therefore, $\Delta(s,a)$ is defined if and only if $\Delta(s',a)$ is defined.
By definition of $\mathcal{B}$, there is an $a$ transition from $\mathbf{s}$ iff $\Delta(s',a)$ is defined iff $\Delta(s,a)$ is defined iff there is an $a$ transition in $\A$.
Thus $\invdef{w'}$ holds.

Now assume both runs are defined for $w'$ and let $s \xrightarrow{a} s'$ in $\A$ and $\mathbf{s} \xrightarrow{a} \mathbf{s'}$ in $\mathcal{B}$.
Let $p \in \Proc$, its state in $\mathbf{s}$ be $(s_p,t_p)$, and its state in $\mathbf{s'}$ be $(s_p',t_p')$.
If $p$ was not part of the communication on $a$, then its state is unchanged and
by definition $\traceparentview{p}{w'}= $ \kl{$\traceparentview{p}{w}$} and
$\traceview{\{p\}}{w'} =$ \kl{$\traceview{\{p\}}{w}$}, so $\invstateparent{w'}$ and
$\invstate{w'}$ are trivially obtained from $\invstateparent{w}$ and
$\invstate{w}$ respectively.

Assume now that $p$ was part of this communication.
Let us first show $\invstate{w'}$.
By definition, $\traceview{\{p\}}{w'} =\traceview{\CA(a)}{w} \cdot a$.
Again, by combining $\invstate{w}$, $\invstateparent{w}$, and Lemma~\ref{lemma:diamview}, we get that $\stateview{\CA(a)}{w} = \diamtree(T_a) = s'$.
Thus $\stateview{\{p\}}{w'} =\Delta(s',a) = t_p'$ and $\invstate{w'}$ holds.
For $\invstateparent{w'}$, there are two cases.
First, assume that the parent of $p$ in $\tree$ is not part of the communication on $a$ (either because $p$ is the root and has no parent, or because its parent does not listen to $a$).
In that case, $\traceparentview{p}{w'} = \kl{\traceparentview{p}{w}}$ by definition and $\kl{\stateparentview{p}{w}}= s_p$ by $\invstateparent{w}$.
Moreover, $p$ must be the root of $T_a$.
Therefore $s_p' = \ifroot{a}{p}{s_p}{\Delta(s_a,a)} = s_p$ and $\invstateparent{w'}$ is satisfied.
Second, assume now that $p$'s parent, say $q$, is part of the communication on $a$.
Then by definition $\traceparentview{p}{w'} = \traceview{\{p\}}{w'} =
\traceview{\CA(a)}{w} \cdot a$. Thus,  $\kl{\stateparentview{p}{w'}} =\Delta(s',a)$ by the same proof as for
$\invstate{w'}$.
Thus $s_p' = \ifroot{a}{p}{s_p}{\Delta(s',a)} = \Delta(s',a)$ and we have
$\invstateparent{w'}$.

\medskip

\noindent \emph{Conclusion of the proof.}
We have shown that the invariants hold for any word $w$.
We use them to show that $\Lang(\A) = \Lang(\mathcal{B})$.
Let $w \in \Alp^*$ be a word, then
$w$ is accepted by $\A$ iff there is a run of $\A$ on $w$ ending in a state $s \in F$.
By $\invdef{w}$, the run of $\A$ on $w$ is defined iff the run of $\mathcal{B}$ on $w$ is defined.
If both runs are undefined, then $w$ is neither in $\Lang(\A)$ nor $\Lang(\mathcal{B})$.
Otherwise, by $\invstateparent{w}$, $\invstate{w}$, and Lemma~\ref{lemma:diamview}, we have that $\diamtree(T) = \traceview{\Proc}{w} = s$.
Then $\mathcal{B}$ accepts $w$ iff $\diamtree(T) = s \in F$ iff $\A$ accepts $w$.
Therefore $\Lang(\A) = \Lang(\mathcal{B})$, which concludes the proof of Theorem~\ref{thm:correctness}.

\subsection{Triangulated dependence alphabets}
\label{sec:monoid}

The complexity of Zielonka's general distribution construction stems from the
need to store information regarding knowledge about other processes and a
time-stamping mechanism that allows to put together these pieces of knowledge \cite{Zielonka87}.
Diekert and Muscholl have shown that in the case that the dependence relation
between the letters in $\Alp$ is a \emph{triangulated graph}, this complex
structure of information can be avoided~\cite{DiekertM96,DR95}.
 
We  observe that tree-like architectures give a fresh view on Diekert and
Muscholl's construction, and an \kl{AA} construction that is exponentially better.



Consider a graph $G=(V,E)$. A sequence of vertices $v_1,\ldots, v_n$ is a cycle if for every $i$ we have $(v_i,v_{i+1})\in E$ and $(v_n,v_1)\in E$.
A graph $G=(V,E)$ is \intro{triangulated} if for every cycle $v_1,\ldots, v_n$ such that $n>3$ there exists a pair $(v_i,v_j) \in E$ such that $1<|j-i|<n-1$.
That is, the cycle must have some chord.

Consider an architecture $\CA$ and let $D=D(\CA)$.
It turns out that an architecture is \kl{tree-like} if and only if the
dependency graph $G=(\Alp,D)$ is \kl{triangulated} and connected. 
It follows that if a dependency graph is triangulated but not connected then its
architecture is forest-like.

\begin{lemma}[\cite{Gavril74}]
    An architecture $\CA$ is tree-like iff $G_D=(\Alp,D)$ is \kl{triangulated} and connected.
    \label{lemma:tree-like iff trinagulaged}
\end{lemma}

In Gavril's terms, given a tree, every graph of subtrees of the tree, where edges correspond to non-empty intersection of the subtrees, is triangulated. In the other direction, every triangulated graph can be represented as the graph of intersections of subtrees of an appropriate tree. 
The part about connectivity is simple to add. 

Note first that the construction of \kl{AA} 
for non-connected
dependency graphs easily reduces to the construction for connected dependency
graphs. To see this assume that 
$\Sigma$ is the disjoint union of two pairwise independent
alphabets $\Sigma_1, \Sigma_2$, so $(a,b) \notin D$ for every $a \in \Sigma_1, b\in \Sigma_2$. 
Consider some $I$-diamond DFA $\A$
over $\Sigma$ and an input word $w \in \Sigma^*$. We can use the diamond property (Lemma~\ref{lem:diam}) to infer from the states reached by $\A$ on the projection $w_1$ respectively $w_2$ of
   the input $w$ on  $\Sigma_1$ and $\Sigma_2$, resp., the state reached by $\A$ on $w \sim_I w_1 w_2$. More formally, assume 
that we constructed an \kl{AA} ${\cal B}_1$ over alphabet $\Sigma_1$ and an
\kl{AA} ${\cal
B}_2$ over alphabet $\Sigma_2$. For each of  ${\cal B}_1$, ${\cal B}_2$ we can suppose
that the global state reached on the respective input determines the state of
$\A$ reached on that word. We  obtain an \kl{AA} ${\cal B}$ by composing ${\cal
B}_1$, ${\cal
B}_2$ in parallel (because they share no letter). The global state reached by
${\cal B}$ on $w
\sim_I w_1 w_2$  determines the two
states $s_1=\Delta(s_0,w_1),s_2=\Delta(s_0,w_2)$ of $\A$, and
$\diam(s_0,s_1,s_2,\Sigma_2)$ says which global states are accepting.

The result we improve on in Theorem~\ref{thm:correctness} is stated in the next
theorem. The exponential size comes from the fact that the construction relies
on transition monoids.

\begin{theorem}[\cite{CMZ93,DiekertM96}]
    Given an $I$-\kl{diamond} DFA $\A$ over a 
    \kl{triangulated} dependency graph $G=(\Alp, D)$, an \kl{AA} 
    ${\cal S}$
    that recognizes the same language can be constructed with $|{\cal
    S}|=n!$, where
    $n=|\A|$.
    \label{theorem:automata for triangulated graphs}
\end{theorem}

\section{Distributed Synthesis}
\label{sec:synthesis}

We extend now the result of decidability of the control problem for asynchronous automata from binary-tree architectures~\cite{muscholl2014distributed}  to \kl{tree-like} architectures.
We first recall the definitions and adapt them to our notations. 

\subsection{Controllers and Control Problem}
Let us fix some \kl{distributed alphabet} $(\Alp,\arch)$.
As in the setting of~\cite{RW87}, $\Alp$ is partitioned into System (\intro {controllable}) actions
$\Alp_\sys$ and Environment (\intro {uncontrollable}) actions $\Alp_\env$.
Remember that $\Alp_p$ denotes the set of $p$-\kl{local actions}, then let $\Alp_p^\sys = \Alp_p \cap \Alp_\sys$ and similarly for $\Alp_p^\env$.
As in~\cite{muscholl2014distributed} we require that all communication actions, i.e., actions shared by at least 2
processes, are uncontrollable, in other words all controllable actions
are local: $\Alp_\sys \subseteq \cup_{p\in \Proc} \Alp_p$. 
This is not a  restriction, as the general setting where communication actions can be controllable reduces to this one, as shown in Proposition 6 of \cite{muscholl2014distributed}.
Furthermore, there should be at least one controllable (thus local) action from
each state. This is of course not a restriction because one can always add
local, controllable self-loops when needed. This assumption will make the construction a bit simpler.

For control it is common to use local acceptance conditions, instead of global
ones, because more systems
are \kl{controllable} with local conditions.  That is, we replace the $\Acc$
part of the definition of \kl{AA} by a set $\{\Acc_p\}_{p \in \Proc}$.
Each $\Acc_p$ is of the form $(F_p, \Omega_p)$ with $F_p \subseteq S_p$ and
$\Omega_p : S_p \to \mathbb{N}$ a local parity function.
The idea is that a finite run is accepted by $p$ if it ends in a state in $F_p$,
and an infinite run is accepted by $p$ if it satisfies the (max) parity
condition for $p$. 
Then a run is accepted by $\A$ if it is accepted by all processes $p$.

\AP We say that a run over $w$ is \intro{maximal} if there is no way to decompose $w$ into $w = u \cdot v$ and no $a \in \Alp$ with $\arch(a) \cap \arch(v) = \emptyset$ such that the run over $u \cdot a \cdot v$ is defined.
In plain words, it means we cannot extend the run on $w$ by some action involving processes
that perform only a finite number of actions in $w$. 

Now we define the notion of \kl {controllers}. For the rest of this section, let us fix some AA $\A = ((S_p)_{p \in \Proc},(s_p^0)_{p \in \Proc}, (\delta_a)_{a \in \Alp}, (\Acc_p)_{p \in \Proc})$.
For simplicity we will not use the general definition of controllers but go
directly to what is actually called covering controllers for $\A$ in~\cite{muscholl2014distributed}, and simply
call them controllers.
It is shown in Lemma 10 of \cite{muscholl2014distributed} that the Control Problem for covering controllers is equivalent to the one for general controllers, and so we focus only on the former.

The intuition for a controller is that it is a machine (described by an AA
without acceptance condition) guiding the System in choosing which \kl{controllable}
action(s) to follow in order to get only runs accepted by $\A$.
On the other hand, as Environment actions are \kl{uncontrollable}, the controller should not be able to restrict them from happening.
To do that, a controller can enrich states of $\A$ by adding extra information and use this to choose a subset of controllable actions that are possible from this state, and enable only this subset of actions.

\AP Formally, a \intro {controller} for $\A$ is itself an AA  $\C = ((C_p)_{p \in \Proc},
(c_p^0)_{p \in \Proc}, (\Delta_a)_{a \in \Alp})$ together with a projection $\pi$
that maps each state $c_p \in C_p$ of $\C$ to a state $s_p \in S_p$ of $\A$ and
that satisfies the following requirements:
\begin{itemize}
    \item for all $a \in \Alp$ with $\arch(a) = \{p_1,\dots,p_k\}$, and $\mathbf{c},\mathbf{c'} \in C_{p_1} \times \dots \times C_{p_k}$, we have that $\Delta_a(\mathbf{c}) = \mathbf{c'}$ implies $\delta_a(\pi(\mathbf{c}))) = \pi(\mathbf{c'})$ (with $\pi$ applied component-wise),
    \item for all $a \in \Alp_\env$, if $\delta_a(\pi(\mathbf{c}))$ is defined then so is $\Delta_a(\mathbf{c})$,
    \item for all $p \in \Proc$, $\pi(c_p^0) = s_p^0$.
\end{itemize}

Thanks to the projection $\pi$,  a \kl {controller} $\C$ inherits the acceptance condition of $\A$. 
\AP A \kl {controller} $\C$ is \intro {winning} for $\A$ if every \kl{maximal} run of $\C$ is accepting.


The \intro{Control Problem}  asks, for a given input $\A$, whether there exists a winning controller $\C$ for $\A$ (and if so to build it explicitly).

\begin{example}
    Consider a server-client setup where one process, the server, communicates with several processes, the clients.
    The server broadcasts to all clients that it wants a task to be solved.
    Each client starts working independently on the task requested, and communicates back to the server when it has finished.
    When the server gets two answers back from its children, it then sends another broadcast telling all clients to stop working.

    Formally, we take $\Proc = \{p,c_1,\dots,c_n\}$ with $p$ the server and $c_1,\dots,c_n$ the clients.
    We have two actions $t_1,t_2$ representing requests for two different \emph{tasks}, $\arch(t_1) = \arch(t_2) = \Proc$. 
    Each process $c_i$ has two local actions $p_1^i,p_2^i$ to \emph{progress} those tasks, $\arch(p_1^i) = \arch(p_2^i) = \{c_i\}$. 
    Then there is a shared action $e_i$ to communicate back to the server that the task has \emph{ended}, $\arch(e_i) = \{c_i,p\}$. 
    Finally, there is an action $r$ that can broadcast to all clients to \emph{reset} their state when the task has been successfully completed, $\arch(r) = \Proc$.
    The architecture is tree-like with the server $p$ at its root, and all clients $c_i$ are the children of $p$.
    The automata of the server and the clients are illustrated in Figure~\ref{example:control}.
    For the acceptance condition, we require that the server visits its initial
    state infinitely often, and that each client performs action $e_i$
    infinitely often. For simplicity of presentation, finite runs are not
    accepted (clients do not crash). 

    Now for the \kl{Control Problem}, we say that all local actions (i.e. actions
    $p_k^i$) are controllable, while all non-local ones are uncontrollable by
    definition.
    The naive controller that only enables the correct progressing action on every process, e.g. $p_1^i$ after receiving a $t_1$ broadcast, is not winning.
    Indeed, since Environment controls which processes get to perform their
    $e_i$, it  could choose the same two processes for every task.
    Then those not chosen never get their turn at performing $e_i$ and thus will not be accepting.
    However we can still build a winning controller in a round-robin manner.
    To do this, we simply add to the state of each $c_i$ a count of how many tasks have been requested since the beginning modulo $n$, so we have $C_{c_i} = S_{c_i} \times \{0,\dots,n-1\}$ and the projection $\pi$ is simply the projection on the first component.
    If the count is $k-1$, then the controllers for processes $c_k$ and $c_{k+1}$ are set to enable only the correct progress action, while other processes enable nothing.
    After that, Environment has no choice but to perform the corresponding $e_k$ and $e_{k+1}$ to keep the run going.
    Then the server sends a reset broadcast, the next task is requested, $k$ is incremented, and each process gets to perform its $e_i$ action in turn, thus this controller is winning.
\end{example}

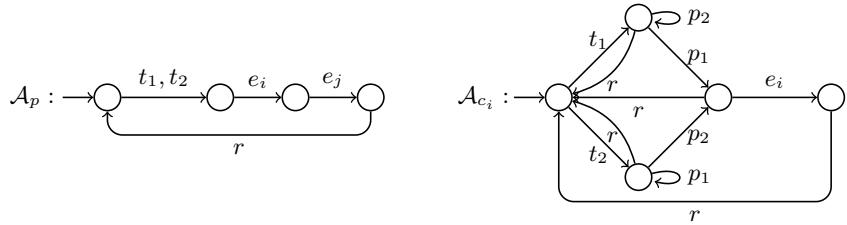
\begin{figure}
\tikzset{every state/.style={minimum size=10pt}}
\begin{center}
  \begin{tikzpicture}[
		auto,
    node distance=1.5cm,
    semithick
    ]
        \node[state,initial left,initial text=,>=stealth] (pinit) at (0,0) {};
        \node[] (p) at (-1,0) {$\A_p:$};
        \node[state] (ptask) at (1.5,0) {};
        \node[state] (p1) at (2.5,0) {};
        \node[state] (p2) at (3.5,0) {};

        \path[->] (pinit) edge [] node [above] {$t_1,t_2$} (ptask);
        \path[->] (ptask) edge [] node [above] {$e_i$} (p1);
        \path[->] (p1) edge [] node [above] {$e_j$} (p2);
        \coordinate (c1) at (3.5,-0.5);
        \coordinate (c2) at (0,-0.5);
        \draw[->,rounded corners=2mm] (p2.south) -- (c1) -- node [below] {$r$} (c2) -- (pinit.south);

        \node[] (c) at (5,0) {$\A_{c_i}:$};
        \node[state,initial left,initial text=] (cinit) at (6,0) {};
        \node[state] (c1) [above right of=cinit] {};
        \node[state] (c2) [below right of=cinit] {};
        \node[state] (cprog) [above right of=c2] {};
        \node[state] (cend) [right of=cprog] {};

        \path[->] (cinit) edge [] node [above] {$t_1$} (c1);
        \path[->] (cinit) edge [] node [below] {$t_2$} (c2);
        \path[->] (c1) edge [] node [right] {$p_1$} (cprog);
        \path[->] (c1) edge [loop right] node [right] {$p_2$} (c1);
        \path[->] (c2) edge [] node [right] {$p_2$} (cprog);
        \path[->] (c2) edge [loop right] node [right] {$p_1$} (c2);
        \path[->] (cprog) edge [] node [above] {$e_i$} (cend);
        \path[->] (c1) edge [bend left] node [below] {$r$} (cinit);
        \path[->] (c2) edge [bend right] node [below] {$r$} (cinit);
        \path[->] (cprog) edge [] node [below] {$r$} (cinit);
        \coordinate (d1) at ($(cend.south)+(0,-1.2)$);
        \coordinate (d2) at ($(cinit.south)+(0,-1.2)$);
        \draw[->,rounded corners=2mm] (cend.south) -- (d1) -- node [below] {$r$} (d2) -- (cinit.south);
  \end{tikzpicture}
\caption{A server-client protocol for distributing and executing tasks.}\label{example:control}
\end{center}
\end{figure}

\subsection{Leaf Process Simulated by Parent}
In \cite{muscholl2014distributed}, it is shown that in an architecture that
forms a tree where all channels are binary, a leaf process can be simulated by
its parent for the purpose of control w.r.t.~local parity conditions.
This is done by abstracting away all local actions of the leaf into a finite
summary, and any synchronization between the leaf and the parent can then be
simulated locally by the parent. This simulation allows to deduce an
upper complexity bound for the \kl{Control Problem} that is a tower of
exponentials in the depth of the tree.

We show here that the proof of~\cite{muscholl2014distributed} extends to \kl{tree-like} architectures.
The key observation is that, in \kl{TCA}, by the continuity condition, any communication involving a leaf and other processes necessarily includes the parent of the leaf as well.
Therefore, if some process relies on the state of the to-be-removed leaf for some transition, it will still be able to access that information after the removal because it will instead be provided by the parent that is simulating the leaf.

%
%

%

For the rest of this subsection we consider an AA $\A$ and a
\kl{TCA} where $\Proc$ is the set of processes, \leaf is a leaf process,
and $p$ is the parent of \leaf.
The first step is to ensure that processes communicate ``frequently''. This will
provide a finite description of what happens locally on leaf \leaf between two
consecutive actions shared by \leaf with other processes. 

\begin{definition}
    An AA $\A$ is \intro {\leaf-short} if there is some bound $B \in \N$ such that in
    every local \leaf-run of $\A$, the number of local actions that \leaf can do
    in a row is
    at most $B$.
\end{definition}

\begin{lemma}[Th.10, \cite{muscholl2014distributed}]
    For every AA $\A$ we can construct an \kl {\leaf-short automaton}
    $\A^\circledS$ such that $\A$ is controllable iff $\A^\circledS$ is
    controllable. The state space of any $q \not=\ell$ does not change from $\A$ to
    $\A^\circledS$; the states of $\ell$ in $\A^\circledS_\ell$ are simple, local
    paths of $\A_\ell$. 
    \label{lemma:l-short automata}
\end{lemma}

Although the proof of Lemma~\ref{lemma:l-short automata} is quite involved, it
does not rely on the communication architecture being a tree. The main idea behind is
that every parity game is equivalent to a finite cycle game.

It follows that to show that a leaf process can be simulated by its
parent (for the purposes of control) it is enough to work with $\A^\circledS$,
so in particular, with an
\kl {\leaf-short} AA. 
From Lemma~\ref{lemma:l-short automata} we know that the bound $B$ can be set to the
number of $\ell$-states in $\A$, so $B=|S_\ell|$ with $S_\ell$ the set of states of
$\ell$ in $\A$. 

We are now ready to show that given $\A$ we can construct an \kl{AA}
$\A^\minusell$ whose set of processes is $\Proc\setminus
\{\leaf\}$ and such that $\A$ is \kl{controllable} iff $\A^\minusell$ is
controllable. Moreover, for every $p'\neq p$ we have $\A^\minusell_{p'}=\A_{p'}$
and the size of $\A^\minusell_p$ is proportional to $|\A_p| \cdot |\Alp_p^\sys| \cdot \mathit{exp}(|\A_\leaf|)$. 

Let $\A = ((S_p)_{p \in \Proc}, (s_p^0)_{p \in \Proc}, (\delta_a)_{a \in \Alp}, (\Acc_p)_{p \in \Proc})$, $\leaf \in \Proc$ a leaf and $p \in \Proc$ its parent, and assume that $\A$ is \leaf-short thanks to Lemma~\ref{lemma:l-short automata}.
With $\Proc^\minusell = \Proc \setminus \{\leaf\}$, we build $\A^\minusell = ((S_p^\minusell)_{p \in \Proc^\minusell}, (s_p^{0\minusell})_{p \in \Proc^\minusell}, (\delta_a^\minusell)_{a \in \Alp^\minusell}, (\Acc_p^\minusell)_{p \in \Proc^\minusell})$.
The new alphabet $\Alp^\minusell$ will be defined shortly after.
For all processes $q \neq p$, we have $S_q^\minusell = S_q$, $s_q^{0\minusell} = s_q^0$, and $\Acc_q^\minusell = \Acc_q$.
For all actions $a \in \Alp$ such that $\arch(a) \cap \{\leaf,p\} = \emptyset$, we have $\delta_a^\minusell = \delta_a$.
That is, all components that are not related to either $p$ or \leaf remain unchanged.

\AP Let us define $S_p^\minusell$. 
To that end, we define \kl {\leaf-local strategies} first.
Remember that $\A$ is \kl{\leaf-short}, so let $B$ be the corresponding bound.
Recall also that $\Sigma_\leaf$ is the set of
$\leaf$-local actions.
For presentation purposes, we write $s \xrightarrow{a_1 \dots a_n} s'$ as a stand in for $s'
= \delta_{a_n}(\dots(\delta_{a_1}(s)))$ when $a_1,\dots,a_n \in \Sigma_\leaf$ and $s,s'$ are
states of $\leaf$.
We write $s\xrightarrow{a_1\cdots a_n}$ when there exists a state $s'$ such that $s \xrightarrow{a_1\dots a_n} s'$. 
An \intro {\leaf-local strategy} from a given state $s_\leaf \in S_\leaf$ is a partial function $f: (\Alp_\leaf)^{\leq B} \to \Alp_\leaf^\sys$ such that whenever $f(w) = a$ then $s_\leaf \xrightarrow{w \cdot a}$ is defined in $\A_\leaf$.
If $s_\leaf \xrightarrow{a} s_\leaf'$ and $f$ is an \leaf-local strategy from $s_\leaf$, then $f_{\mid a}$ is the \leaf-local strategy from $s_\leaf'$ defined as $f_{\mid a}(w) = f(a \cdot w)$.
With this defined, a state in $S_p^\minusell$ of the parent $p$ of $\leaf$ is of one of 3 types:
\[(s_p,s_\leaf), (s_p,s_\leaf,f), (s_p,a,s_\leaf,f)\]
where $s_p \in S_p$, $s_\leaf \in S_\leaf$ are the simulated states of $p$ and
\leaf respectively, $f$ is an \kl{\leaf-local strategy} from $s_\leaf$, and $a
\in \Alp_p^\sys$ is a controllable local action of the parent $p$.
Since $\A$ is \leaf-short, \leaf-local strategies can be seen simply as paths of $\A_\leaf$ without repetitions, and thus there are only $\mathit{exp}(|\A_\leaf|)$ many of them.
Finally, we let $s_p^{0\minusell} = (s_p^0, s_\leaf^0)$ be the initial state of $p$. 

We now precisely define the new alphabet $\Alp^\minusell$ before defining the new transitions.
All actions that were previously \leaf-local are changed into uncontrollable $p$-local actions, regardless of whether they were controllable or not in $\Alp$.
Similarly, all actions that were $p$-local are now uncontrollable in $\Alp^\minusell$, but for each $p$-local action $a \in \Alp_p^\sys$ that was controllable we add a new action $\choice(a)$ in $\Alp^\minusell$ that is also $p$-local and controllable.
We also add new actions $\choice(f)$ that are $p$-local and controllable for each \kl{\leaf-local strategy} $f$.
Any non-local action that included \leaf in $\Alp$ remains in $\Alp^\minusell$ but with \leaf removed from its domain.
All remaining actions in $\Alp$, i.e. those that are local to some process not in $\{p,\leaf\}$ or those that are non-local and do not intersect with \leaf, remain unchanged.

Let us now define transitions over actions involving $p$.
First, from a state of $p$ of the form $(s_p,s_\leaf)$, there is only one kind of transition available
\begin{align*}
(s_p,s_\leaf) \xrightarrow{\choice(f)} (s_p,s_\leaf,f) \text{ for all \kl{\leaf-local strategies} $f$ from $s_\leaf$ }
\end{align*}
where System fixes a local strategy for \leaf.
Then from these states there is again only one kind of transition possible
\begin{align*}
(s_p,s_\leaf,f) \xrightarrow{\choice(a)} (s_p,a,s_\leaf,f) \text{ for all $a \in \Alp_p^\sys$ enabled from $s_p$}   
\end{align*}
where System chooses one $p$-local controllable (in $\Alp_p$)
action\footnote{The assumption that in each state there is at least one enabled
controllable action is used here.}, provided there is a transition from $s_p$ with this
action.
Then there are multiple choices from here. The two actions 
\begin{align*}
&(s_p,a,s_\leaf,f) \xrightarrow{b_p} (s_p', s_\leaf, f) \text{ if } (s_p,s_p') \in \delta_{b_p} \text{ and $b_p=a$ or $b_p \in \Alp_p^\env$}\\
&(s_p,a,s_\leaf,f) \xrightarrow{b_\leaf} (s_p,a,s_\leaf',f_{\mid b_\leaf}) \text{ if $(s_\leaf,s_\leaf') \in \delta_{b_\leaf}$ and $b_\leaf$ is enabled by $f$} 
\end{align*}
are $p$-local actions (all uncontrollable) which either progress $p$ by its pre-selected (previously) controllable action $a$, progress $p$ by any local uncontrollable action $b_p$, or progress \leaf by any action enabled by the pre-selected strategy $f$, i.e. any uncontrollable \leaf-local action or the one controllable action output by $f$ on $\varepsilon$.
Then there are also non-local actions where $p$ communicates with other processes.
\begin{align*}
&((s_p,a,s_\leaf,f),s_{q_1},\dots,s_{q_k}) \xrightarrow{b_{p+}} ((s_p',s_\leaf,f),s_{q_1}',\dots,s_{q_k}') \text{ if }\\
\notag&\quad p \in \arch(b_{p+}), \leaf \notin \arch(b_{p+}), ((s_p,s_{q_1},\dots,s_{q_k}),(s_p',s_{q_1}',\dots,s_{q_k}')) \in \delta_{b_{p+}}\\
&((s_p,a,s_\leaf,f),s_{q_1},\dots,s_{q_k}) \xrightarrow{b_{p\leaf+}} ((s_p',s_\leaf'),s_{q_1}',\dots,s_{q_k}') \text{ if }\\
\notag&\quad \leaf,p \in \arch(b_{p\leaf+}), ((s_p,s_\leaf,s_{q_1},\dots,s_{q_k}),(s_p',s_\leaf',s_{q_1}',\dots,s_{q_k}')) \in \delta_{b_{p\leaf+}}
\end{align*}
The first kind is for when \leaf is excluded from the communication, so while $s_p$ gets updated to a new state and the choice of $a$ is reset, $s_\leaf$ and $f$ remain unchanged.
The second kind occurs when \leaf is part of the communication, so both $s_p$ and $s_\leaf$ are read and updated, and the choices of $a$ and $f$ are reset.
Note that $k=0$ covers the case where a communication involving only $p$ and \leaf occurs.
Also, since we assumed that the architecture is tree-like, it is not possible for a communication to involve \leaf and other processes without $p$ being included, which is why this case is not covered.
Note that this is the only major difference with the original construction from \cite{muscholl2014distributed}, and it is the reason that the proofs of correctness are mostly similar.

The last  part of $\A^\minusell$ is the acceptance conditions $\Acc_p^\minusell = (F_p^\minusell,\Omega_p^\minusell)$.
For $F_p^\minusell$ it is simple: any state of the form $(s_p,s_\leaf) \in F_p
\times F_\leaf$, and any $(s_p,s_\leaf,f),(s_p,a,s_\leaf,f)$ where $s_p \in F_p$
and any possible maximal
continuation respecting $f$ from $s_\leaf$ eventually reaches a state in
$F_\leaf$, are in this set. 
Defining $\Omega_p^\minusell$ is trickier. First,  an infinite run of $p$ must satisfy the parity condition of $\Omega_p$ when restricted to the $s_p$ component of the state.
Then there are two cases.
If the projection on the $s_\leaf$ component is infinite and satisfies
$\Omega_\leaf$ then the run must be accepted, which amounts to a conjunction of
two parity conditions. This can be translated into a single parity condition,
and we will discuss the complexity in Section~\ref{sec:wrap-up-control}.
Otherwise, we require that for the last $f$ and $s_\leaf$ reached (uniquely
defined as the only components seen infinitely often in the run), any possible
continuation respecting $f$ from $s_\leaf$ eventually reaches a state in
$F_\leaf$. This amounts to a single parity condition, the one on $p$. 

Note that we do not define what happens when the run on \leaf is infinite but the run on $p$ is finite, because this is not compatible with the \leaf-short assumption.

Now that $\A^\minusell$ is defined, we are ready to prove the following theorem (proof in appendix).

\begin{apxtheoremrep}
    For every \kl {\leaf-short} AA $\A$ with local acceptance condition:
    there is a \kl {winning controller} for $\A$ iff there is a winning controller for $\A^\minusell$.
    \label{theorem:remove one leaf}
\end{apxtheoremrep}

\begin{proof}
    \fbox{$\Rightarrow$}
    Let $\C = ((C_p)_{p \in \Proc}, (c_p^0)_{p \in \Proc}, (\Delta_a)_{a \in \Alp})$ be a winning controller for $\A$ with associated projection $\pi$.
    Let $c \in C_\leaf$, we build an \kl{\leaf-local strategy} from $\pi(c)$ called $f_c$ as follows.
    Let $c \xrightarrow{a_1} c_1 \xrightarrow{a_2} \dots \xrightarrow{a_k} c_k$ be a run in $\C_\leaf$ with $a_1,\dots,a_k \in \Alp_\leaf$, and let $a$ be one controllable action enabled from $c_k$ (if there are more than one, pick any).
    Then we define $f_c(a_1 \dots a_k) = a$.
    Since $\A$ is \leaf-short, $f_c$ is a well-defined \leaf-local strategy.
    
    We build the controller $\C^\minusell = ((C_p^\minusell)_{p \in \Proc^\minusell}, (c_p^{0\minusell})_{p \in \Proc^\minusell}, (\Delta_a^\minusell)_{a \in \Alp^\minusell})$ for $\A^\minusell$.
    Components not related to $p$ are unchanged: $C_q^\minusell = C_q$, $c_q^{0\minusell} = c_q^0$ for all $q \neq p$, and $\Delta_a^\minusell = \Delta_a$ for all $a \in \Alp^\minusell$ such that $p \notin \arch(a)$.
    States in $C_p^\minusell$ are of the form $(c_p,c_\leaf)$,
    $(c_p,c_\leaf,f)$, or $(c_p,a,c_\leaf,f)$ where $c_p \in C_p$, $c_\leaf \in
    C_\leaf$, $a \in \Alp_p$, and $f$ is a \leaf-local strategy from
    $\pi(c_\leaf) $.
    Its initial state is $c_p^{0\minusell} = (c_p^0,c_\leaf^0)$.
    The associated projection $\pi^\minusell$ is derived from $\pi$: $\pi^\minusell(c_p,c_\leaf) = (\pi(c_p),\pi(c_\leaf))$, and so on.
    
    Finally we describe the transitions associated with component $p$.
    The only controllable action enabled from state $(c_p,c_\leaf)$ is the action $\choice(f)$ where $f = f_{c_\leaf}$ is the \leaf-local strategy extracted from $C_\leaf$ as defined above, leading to state $(c_p,c_\leaf,f)$.
    The only controllable action enabled from $(c_p,c_\leaf,f)$ is the action
    $\choice(a)$ for one action $a \in \Alp_p^\sys$ enabled in $c_p$, which
    leads to state $(c_p,a,c_\leaf,f)$. If there are several such $a$, uniquely
    pick one.
    If there is none, assign some arbitrarily picked action\footnote{This helps to simplify the definition of transitions,
    avoiding case distinction. In this case, the only possible continuations are
    either non-controllable actions of $p,\leaf$, or applying the strategy $f$.} $a_0\in
    \Sigma_p$.
    All other actions are uncontrollable, and follow the structure of $\A^\minusell$:
    \begin{itemize}
        \item $(c_p,a,c_\leaf,f) \xrightarrow{b_p} (c_p',c_\leaf,f)$ if $c_p \xrightarrow{b_p} c_p'$ and either $b_p=a$ or $b_p \in \Alp_p^\env$ executes either the previously picked controllable $a$ or any uncontrollable $p$-local action,
        \item $(c_p,a,c_\leaf,f) \xrightarrow{b_\leaf} (c_p,a,c_\leaf',f_{\mid b_\leaf})$ if $c_\leaf \xrightarrow{b_\leaf} c_\leaf'$ and either $b_\leaf=f(\varepsilon)$ or $b_\leaf \in \Alp_\leaf^\env$ executes an enabled \leaf-local action,
        \item $((c_p,a,c_\leaf,f),c_{q_1},\dots,c_{q_k}) \xrightarrow{b_{p+}} ((c_p',c_\leaf,f),c_{q_1}',\dots,c_{q_k}')$ if $\Delta_{b_{p+}}(c_p,c_{q_1},\dots,c_{q_k}) = \\(c_p',c_{q_1}',\dots,c_{q_k}')$ is a communication between $p$ and other processes not including \leaf,
        \item $((c_p,a,c_\leaf,f),c_{q_1},\dots,c_{q_k}) \xrightarrow{b_{p\leaf+}} ((c_p',c_\leaf'),c_{q_1}',\dots,c_{q_k}')$ if $\Delta_{b_{p\leaf+}}(c_p,c_\leaf,c_{q_1},\dots,c_{q_k}) = \\(c_p',c_\leaf',c_{q_1}',\dots,c_{q_k}')$ is a communication between $p$ and other processes including \leaf.
    \end{itemize}
    It is easy to see that $\C^\minusell$ is indeed a controller for $\A^\minusell$ and that $\pi^\minusell$ satisfies the required conditions since by definition $\C^\minusell$ mimics the structure of $\A^\minusell$.

    Finally, we show that $\C^\minusell$ is winning for $\A^\minusell$.
    Let $\rho^\minusell$ be a maximal run of $\C^\minusell$.
    From it, we build the corresponding run $\rho$  of $\C$ by removing from
    $\rho^\minusell$ transitions with actions of the form $\choice(f)$ and
    $\choice(a)$, hiding extra components $a$ and $f$ in the states, and
    decoupling components $c_p$ and $c_\leaf$.
    One can prove by induction that $\rho$ is indeed a run of $\C$.

    
    It is possible that $\rho$ is not maximal when the projection on $\leaf$ is
    finite, ending say with $c_\leaf,f$.  In this case we take any possible local continuation for $\leaf$
    according to $f$. So we can assume that the run $\rho$  is maximal. 

    We show that $\rho^\minusell$ is accepting.
    We know that $\rho$ is maximal and $\C$ is winning, so all components are accepting.
    This directly implies that for all $q \neq p$, the run $\rho^\minusell$ over component $q$ is accepting, whether it is finite or infinite.
    Only component $p$ remains.
    If the run over $p$ is finite, it ends in a state of the form $(c_p,a,c_\leaf,f)$.
    This means that in $\rho$, components $p$ and \leaf end in states $c_p$ and
    $c_\leaf$. First we note that  $\pi(c_p) \in F_p$. Since $\C_\leaf$ is also
    winning, all maximal continuations of $\leaf$ from $c_\leaf$ according to
    $f$ reach $\pi^{-1}(F_\leaf)$. 
    Therefore $\pi^\minusell(c_p,a,c_\leaf,f) \in F_p^\minusell$.
    If the run over $p$ is infinite, either there are infinitely many actions affecting \leaf or not.
    In the former case, then in $\rho$ both runs over $p$ and \leaf are infinite, and belong to $\Omega_p$ and $\Omega_\leaf$ respectively.
    Then by definition in $\rho^\minusell$ the run over $p$ satisfies $\Omega_p^\minusell$.
    In the latter case, in $\rho$ the run over $p$ is infinite and belongs to $\Omega_p$, but for \leaf the run is finite and ends in some state $c_\leaf \in \pi^{-1}(F_\leaf)$ with no possible continuation, otherwise contradicting the maximality assumption.
    Then $\Omega_p^\minusell$ is satisfied, which ends this direction of the proof.
        
    \fbox{$\Leftarrow$}
    Let $\C^\minusell = ((C_p^\minusell)_{p \in \Proc^\minusell}, (c_p^{0\minusell})_{p \in \Proc^\minusell}, (\Delta_a^\minusell)_{a \in \Alp^\minusell})$ be a winning controller for $\A^\minusell$ with associated projection $\pi^\minusell$.
    We want to build a winning controller $\C$ for $\A$.

    Let us first give the intuition for $\C$.
    We can assume that $\C^\minusell$ only enables at most one controllable
    action at a time (such a controller is more restrictive, so if it was
    winning  before, it remains winning).
    Since states whose projection is of the form $(s_p,s_\leaf)$ or $(s_p,s_\leaf,f)$ only have outgoing transitions with controllable actions, this means that from those states there is a forced path to a \emph{true} state of the form $(s_p,a,s_\leaf,f)$.
    We will ignore the former intermediary states and only focus on the true ones.
    The goal is to decouple the component $C_p^\minusell$ into $C_p$ and $C_\leaf$.
    Each of them will hold a copy of the state of $C_p^\minusell$, and $C_\leaf$ will additionally hold the sequence of \leaf-local actions made since the last communication with $p$.
    Both $p$ and \leaf will progress independently, respecting the choices of $a$ and $f$ respectively.
    When a communication between them (and possibly other processes) happens, they will combine the state of $p$ and the sequence of actions of \leaf to get back the actual state that would have been reached in $C_p^\minusell$, and then update accordingly.
    This works because the progress of $p$ and \leaf is independent due to the tree-like structure, meaning that we can combine them sequentially and get the same result as any other order.

    Formally, we define true states of component $p$ as the set 
    \[\TS = \{c_p \in C_p^\minusell \mid \pi^\minusell(c_p) \text{ is of the form } (s_p,a,s_\leaf,f)\}\]
    and for any state $c_p \notin \TS$ we let $\ts(c_p)$ be the unique true state reachable from $c_p$ with either one $\choice(a)$ transition, or two $\choice(f), \choice(a)$ transitions.

    We build the controller $\C = ((C_p)_{p \in \Proc}, (c_p^0)_{p \in \Proc}, (\Delta_a)_{a \in \Alp})$ for $\A$.
    As before, all components unrelated to $p$ and \leaf are the same as in $\C^\minusell$.
    States of $p$ are $C_p = \TS$, and states of \leaf are $C_\leaf = \TS \times \Alp_\leaf^{\leq B}$.
    Initial states are $\ts(c_p^{0\minusell})$ and $(\ts(c_p^{0\minusell}),\varepsilon)$ respectively.
    Projection $\pi$ is defined as follows: let $c \in \TS$ and
    $\pi^\minusell(c) = (s_p,a,s_\leaf,f)$, then for $p$ we set $\pi(c) = s_p$
    and for \leaf we set $\pi((c,w))=s'_\leaf$ with  $s_\leaf \stackrel{w}{\to} s'_\leaf$.
    There are four different types of transitions that involve $p$ or \leaf.
    \begin{itemize}
        \item For $p$-local transitions $b \in \Alp_p$, we have
        $\Delta_{b}(c_p) = \ts(\Delta_{b}^\minusell(c_p))$ i.e., we progress
        as in $\C_p^\minusell$ but skip non-true states.
        \item For \leaf-local transitions $b \in \Alp_\leaf$, we have
        $\Delta_{b}((c_p,w)) = (c_p,wb)$ if $c_p
        \stackrel{wb}{\to}$ in $C_p^\minusell$,
        i.e., we only record the action made in
        the second component of the state.
        Note that $|wb| \le B$ because of the projection of states of
        $\C^\minusell$ onto states of $\A^\minusell$, and because of the
        $B$-short strategies $f$.
        \item Let $b$ involve $p$ and $q_1,\dots,q_k$, but not \leaf, and assume
        that $\Delta_b^\minusell(c_p,c_{q_1},\dots,c_{q_k}) =
        (c_p',c_{q_1}',\dots,c_{q_k}')$. Then we set
        $\Delta_b(c_p,c_{q_1},\dots,c_{q_k}) =
        (\ts(c_p'),c_{q_1}',\dots,c_{q_k}')$ i.e., same as in the first case.
        \item Let $b$ involve $p$, \leaf, and $q_1,\dots,q_k$, and assume that
        $\Delta_b^\minusell(c_p^3,c_{q_1},\dots,c_{q_k}) =
        (c_p',c_{q_1}',\dots,c_{q_k}')$.  Assume also that $w \in
        \Sigma^*_\leaf$ is a sequence of local $\leaf$-actions such that $c_p^1
        \stackrel{w}{\to} c_p^3$  in $C_p^\minusell$.
        
         We set $\Delta_b((c_p^2,w),c_p^1,c_{q_1},\dots,c_{q_k}) =
         ((\ts(c_p'),\varepsilon),\ts(c_p'),c_{q_1}',\dots,c_{q_k}')$ in ${\cal C}$.
         That is, first we combine the state $c_p^1$ reached in the
         $p$-component (from $c_p^2$)
        with the sequence $w$ accumulated in \leaf to obtain the actual state
        $c_p^3 \in C_p^\minusell$, followed by the $b$-transition in
        $\C^\minusell$. From the new state $c'_p$ we compute the new true state, and set the states of $p$ and \leaf
        to the new true state while also resetting the sequence of \leaf-local
        transitions in the \leaf component.
    \end{itemize}
    Showing that $\C$ is indeed a controller of $\A$ is done by checking each requirement on each type of transition.
    The first three types pose no problem simply by the way $\C$ is defined.
    For the fourth type it relies on the observation that the $p$-component and the \leaf-component are progressing independently between communications.
    Thus their progress can be combined in the way described above to obtain back an actual run of $\C^\minusell$.

    Given any run in $\C$, one can build an equivalent \emph{well-ordered} run where in between two actions involving both $p$ and \leaf, we push all \leaf-local actions to the end.
    This means we can focus on well-ordered runs of $\C$ only.
    Then, any well-ordered run $\rho$ of $\C$ can be translated to a sequence $\rho^\minusell$ by simply adding the (uniquely defined) missing actions $\choice(a)$ and $\choice(f)$ in between every two true states, if needed.
    A case analysis on the transitions shows the following by induction on the length of $\rho$:
    \begin{enumerate}
        \item $\rho^\minusell$ is a well-defined run of $\C^\minusell$,
        \item for processes $q \notin \{p,\leaf\}$, the state of $q$ at the end of $\rho$ is the same as in $\rho^\minusell$,
        \item for processes $p$ and \leaf, we have that the $\pi$ projection of the state of $p$ (respectively \leaf) at the end of $\rho$ 
        is the same as the $p$ (respectively \leaf) component of the $\pi^\minusell$ projection of the state of $p$ at the end of $\rho^\minusell$,
        \item $\rho^\minusell$ is maximal if $\rho$ is maximal.
    \end{enumerate}
    Combining those properties and the fact that $C^\minusell$ is a winning controller for $\A^\minusell$ leads to $\C$ being a winning controller for $\A$, which concludes the proof.
\end{proof}

\subsection{Solving the Control Problem on Tree-like Architectures}\label{sec:wrap-up-control}

Theorem~\ref{theorem:remove one leaf} generalizes the leaf-elimination procedure
of \cite{muscholl2014distributed} from binary channels to \kl{tree-like}
architectures. Iterating this step reduces the distributed control problem to one
involving a single process, so to a usual parity game. It is argued in
\cite{muscholl2014distributed}  that the size of the resulting parity game is
$\text{Tower}_d(n)$, where $n$ is the size of the AA and $d \ge 0$ the depth of the
process tree. A matching lower bound for the control problem is provided in
\cite{GGMW13}. For the upper bound
 we provide below a detailed argument based on the following lemma:

\begin{lemma}
    Given $n+1$ parity conditions $p_0$ and $p_1,\ldots, p_n$, there is a deterministic parity automaton with $O\left ( (p_0+p)\cdot \frac{p! \cdot \left (p_0+1\right )^p}{\left (\prod_{i>0}\left (p_i!\right )\right )}\right )$ states and $p_0+p$ priorities, where $p=\sum_{i>1}p_i$, reading tuples of $n+1$ priorities and accepting iff all $n+1$ parity conditions are accepting.
    \label{lemma:parity conversion}
\end{lemma}

\begin{proof}
    We use a variant of the index appearance record construction \cite{Safra92,Loding99} to convert $n+1$ parity conditions to one condition.
    We start by using the index appearance record to convert the conditions $p_1,\ldots, p_n$ to one parity condition with $p$ priorities.
    This is a straightforward application of the index appearance record, noticing that the order within each parity condition is always maintained. Thus, we keep a permutation of $p$ indices, but where all the indices belonging to a single parity condition are always ordered in decreasing order.
    Thus, there are $\frac{p!}{\left (\prod_{i>0}\left (p_i!\right )\right )}$ such permutations.
    It remains to keep the respective order of the elements in $p_0$ among the sequence of $p$ indices.
    As the elements of $p_0$ themselves are ordered in decreasing order, it is sufficient to keep track of how many of them appear before every element in the sequence of $p$ indices.
    Thus, we multiply the number of options by $(p_0+1)^p$.
    Thus, $\frac{p! \cdot \left (p_0+1\right )^p}{\left (\prod_{i>0}\left (p_i!\right )\right )}$ possible values represent all possible sequences of $p_0+p$ indices, where the indices of each parity condition itself are ordered in decreasing order.
    Every transition of the created parity automaton needs to signal a priority between $p_0+p$ and $1$.
    For that, we keep a record of what was the left-most index that is visited in a transition from one $p_0+p$ sequence to its successor.
    This brings the total number of states to the stated in the Lemma.
    We note that $p_0$ does not appear in the exponent.
    \qed
\end{proof}

We note that the complexity of solving games where the winning condition is a
conjunction of parity conditions is known to be co-NP-complete
\cite{ChatterjeeHP07}. Thus, a much simpler  conversion to parity games than the
one described in the 
previous lemma is
not possible.

    We also note that the upper bound of the previous construction
    \cite{muscholl2014distributed} was assuming that the acceptance condition is
    not part of the input.

\begin{corollary}\label{cor:upper}
    The \kl{Control Problem} for asynchronous automata over tree-like architectures is decidable in $d$-EXPTIME where $d$ is the depth of the tree.
\end{corollary}

\begin{proof}
    We apply Theorem~\ref{theorem:remove one leaf} iteratively. 
    Choose a parent $p$ with leaves $\ell_1,\ldots, \ell_n$. 
    Apply the theorem to each one of the leaves $\ell_1,\ldots, \ell_n$.
    We end up with an equi-realizable control problem, where $\ell_1,\ldots, \ell_n$ are missing and $p$ is replaced by $p^{\setminus \ell_1,\ldots, \ell_n}$.
    The number of states of $p^{\setminus \ell_1,\ldots, \ell_n}$ is proportional to $|S_p|\cdot |\Sigma|^n \cdot \prod_{i>0} exp(|S_{\ell_i}|)$.
    Furthermore, the infinitary condition of $p^{\setminus \ell_1,\ldots, \ell_n}$ is a conjunction of $n+1$ parity conditions.
    By Lemma~\ref{lemma:parity conversion}, if the number of parities of $p$ and $\ell_1,\ldots, \ell_n$ are $p_0,\ldots, p_n$, respectively, then by multiplying the number of states of $p^{\setminus \ell_1,\ldots, \ell_n}$ by $O\left ( (p_0+p)\cdot \frac{p! \cdot \left 
    (p_0+1\right )^p}{\left (\prod_{i>0}\left (p_i!\right )\right )}\right )$ we get a parity condition with $\sum_{i\geq 0}p_i$ priorities.

    It follows that by eliminating leaves level by level, we get an asynchronous
    automaton whose size is a tower of exponentials of height $d$, where $d$ is
    the depth of the eliminated tree, and whose parity index is the sum of all
    parity indices of all processes in the tree.

    When the entire tree is reduced to a single process, we get a control problem for one process and the corollary follows.
    \qed
\end{proof}

The $d$-EXPTIME complexity bound stated in Corollary~\ref{cor:upper} is matched
by a result from~\cite{GGMW13} showing that how to reduce acceptance of a $d$-EXPSPACE
bounded Turing machine to a \kl{Control Problem}  on a \emph{tree architecture} of
depth $O(d)$ (although not exactly $d$, so some gap remains).

\section{Conclusion}

We revisited simpler constructions of Zielonka asynchronous automata, and the 
control problem, in the case of tree-like process architectures. We showed that
the quadratic construction of Zielonka automata generalizes from trees to
tree-like architectures, and also showed how this improves on  a previous, exponential construction of
asynchronous automata for triangulated dependence alphabets. Finally we showed that distributed
controller synthesis  generalizes smoothly to tree-like architectures, too.

While causal synchronisation like in the Zielonka model may appear far from real systems,
it represents a first step to handle communication. 
Communication has different flavor than rendez-vous mechanisms like in the Zielonka model, e.g.~it lacks the symmetry of Zielonka synchronisations since receivers are better informed than senders. 
At the same time, causal dependencies, like those arising in the Zielonka model, are important and can be approximated by practical forms of communication (e.g., local broadcast). 
In such a context, we believe that tree-like process architectures may provide more potential for
future applications involving communication, than tree architectures.
Indeed, tree-like architectures allow for real multi-party communication as channels 
are not restricted to binary connections.
It adds the requirement to inform all the ``region'' affected by a communication,
suggesting possible uses where interaction is connected to physical location. 
In our future work, we are working on an implementation of a 
communication infrastructure that supports a relaxed version of Zielonka type synchronization.
Our framework ensures that actions happen in the same partial order as in the theoretical (Zielonka)
model.
The communication framework currently enacts asynchronous automata that are modelled by hand.
However, if the communication architecture is tree-like we could use an
automated distribution of centrally 
crafted plans based on similar ideas as in this paper.
At a theoretical level, we plan to investigate which central plans can be distributed into a tree-like communication framework with bounded communication channels.

\clearpage
\bibliographystyle{plainurl}
\bibliography{bib}

\end{document}